\documentclass[a4paper,reqno,11pt]{amsart}
\usepackage{amsmath, graphicx}
\usepackage{amstext,amsfonts,amsbsy,eucal,amssymb}
\usepackage{endnotes}

\numberwithin{equation}{section}

\newtheorem{Theorem}{Theorem}[section]
\newtheorem{Lemma}[Theorem]{Lemma}
\newtheorem{Prop}[Theorem]{Proposition}

\newtheorem{Cor}[Theorem]{Corollary}

\theoremstyle{definition}
\newtheorem{Definition}[Theorem]{Definition}
\newtheorem{Example}[Theorem]{Example}
\newtheorem{Remark}[Theorem]{Remark}

\newcommand{\R}{\operatorname{\mathbb{R}}}
\newcommand{\Z}{\operatorname{\mathbb{Z}}}
\newcommand{\C}{\operatorname{\mathbb{C}}}
\newcommand{\Q}{\operatorname{\mathbb{Q}}}

\newcommand{\Div}{\operatorname{\mathrm{Div}}}
\newcommand{\Jac}{\operatorname{\mathrm{Jac}}}
\newcommand{\Pic}{\operatorname{\mathrm{Pic}}}

\newcommand{\val}{\operatorname{\mathrm{val}}}

\begin{document}

\title{Tropical curves and integrable piecewise linear maps}

\author{Rei Inoue}
\address{Chiba University,
\newline\phantom{iii} 1-33 Yayoi-cho, Inage,
Chiba 263-8522, Japan}
\email{reiiy@math.s.chiba-u.ac.jp}

\author{Shinsuke Iwao}
\address{Rikkyo University,
\newline\phantom{iii} 3-34-1 Nishi-Ikebukuro Toshima-ku, 
Tokyo 171-8501, Japan}
\email{iwao@rikkyo.ac.jp}

\begin{abstract}
We present applications of tropical geometry
to some integrable piecewise-linear maps,
based on the lecture given by one of the authors (R. I.) at the workshop
``Tropical Geometry and Integrable Systems" 
(University of Glasgow, July 2011), 
and on some new results obtained afterward. 
After a brief review on tropical curve theory,
we study the spectral curves and the isolevel sets of 
the tropical periodic Toda lattice and the periodic Box-ball system.
\end{abstract}




\thanks{
R.~I. is partially supported by Grant-in-Aid for Young Scientists (B)
(22740111).
S.~I. is supported by Grant-in-Aid for JSPS Fellows
(21-1939).}

\maketitle

\section{Introduction}

\subsection{Background --- integrable systems and algebraic geometry}
\label{sec:1-1}
Let us show an example of the remarkable application  of 
complex algebraic geometry to integrable systems.

Fix $N \in \Z_{>1}$, and let $\Z_N$ be the quotient ring $\Z / N\Z$.
The $N$-periodic Toda lattice equation is a famous integrable system given by
\begin{align}\label{originalToda}
a_n' = b_n - b_{n-1}, \qquad 
b_n' = b_n(a_{n+1} - a_n),
\end{align}
on the {\it phase space} $M=\{(a_n,b_n)_{n \in \Z_N}\} \simeq \C^{2N}$.
Here we write $a_n'$ for a derivation of $a_n = a_n(t)$ by the time $t$.
\footnote{The original form of the Toda lattice equation is 
$x_n'' 
=  \mathrm{e}^{x_{n+1} - x_n} - \mathrm{e}^{x_{n} - x_{n-1}}
$.
One obtains the above form via the transformation:
$a_n= x_n'$, $b_n = \mathrm{e}^{x_{n+1} - x_n}$.}
To solve this equation, we use an important property that
there are $N+1$ algebraically independent polynomial functions 
$h_j~(j=1,\ldots,N+1)$ on $M$,
which are conserved by the equation.
Fix $c = (c_1,\ldots,c_{N+1}) \in \C^{N+1}$
and define the subset of $M$ by 
$$
M_c = \{ m \in M ~|~ h_j(m) = c_j ~(j=1,\ldots,N+1) \}.
$$
This is called the {\it isolevel set} 
invariant under the time evolution.
By the definition, $M_c$ is an algebraic variety.
Since the holomorphic function on $M_c$ corresponds to the solution,
we want to know what kind of algebraic variety it is.
Let $\gamma_c$ be the algebraic curve given by
\begin{align}\label{Toda-curve}
y^2 + y (x^N + c_1 x^{N-1} + \cdots + x c_{N-1} + c_N) + c_{N+1} = 0.
\end{align}
The curve $\gamma_c$ is called the {\it spectral curve},
which is also invariant under the time evolution.
When $\gamma_c$ is smooth
(i.e. $\gamma_c$ is a hyperelliptic curve of genus $N-1$),
we have the followings \cite{DT76,KvM75}:
\begin{itemize}
\item[(i)]
The isolevel set $M_c$ is isomorphic to an affine part of the Jacobian
$\Jac(\gamma_c) $ of $\gamma_c$.

\item[(ii)]
The solution is written in terms of the corresponding 
Riemann's theta function. Moreover, the flow of the equation 
is linearized on $\Jac(\gamma_c)$.  
\end{itemize}

Let us explain more detail in the case of $N=2$.
With $c = (c_1, c_2, c_3) \in \C^3$, we fix the three conserved functions as
$$
h_1=a_1+a_2 = c_1, \qquad 
h_2 = a_1a_2-b_1-b_2 =c_2, \qquad
h_3 = b_1 b_2 =c_3.
$$
By erasing $a_2$ and $b_1$ in these relations, we obtain 
$b_2^2 + b_2(a_1^2-a_1c_1+c_2)+c_3 =0$.
This is nothing but the defining equation of $\gamma_c$
via $(x,y)=(-a_1,b_2)$.
The map $\phi : M_c \to \Jac(\gamma_c)$ is a composition of 
$$
\begin{matrix}
M_c & \to & \gamma_c & \stackrel{\rm AJ}{\to} & \Jac(\gamma_c)
\\ 
(a_1, a_2,b_1,b_2) & \mapsto & (-a_1,b_2):= P & \mapsto 
&\displaystyle{\int_{P_0}^P \omega}
\end{matrix}~,
$$
where AJ is the Abel-Jacobi map with a base point $P_0 \in \gamma_c$
and the holomorphic one form $\omega$ on $\gamma_c$.
The image of $\phi$ is $\Jac(\gamma_c)$ minus two points corresponding
to the infinity points of $\gamma_c$.

As one sees in this example,
the key to solve the equation is to find the spectral curve whose Jacobian 
is related to the isolevel set.

\subsection{Tropicalization}

Let $K$ be an algebraic closed field
with a valuation $\val : K \setminus \{0\} \to \R$ as 
$$
\val(ab) = \val(a)+\val(b), 
\qquad
\val(a+b) \geq \min[\val(a),~\val(b)], 
$$
for $a,b \in K$.
The {\it tropicalization} is the map from $K$ to $\R$ by this valuation.
For a polynomial 
$f = \sum_{m \in I} c_m x^m \in K[x_1,\ldots,x_n]$ 
where $I$ is a finite subset of $(\Z_{\geq 0})^n$,
we define its {\it tropicalization} as
$$
F = \min_{m \in I} [\val(c_m) + m \cdot X],
$$
where we set $m=(m_1,\ldots,m_n)$, $x^m = \prod_i x_i^{m_i}$ and 
$m \cdot X = \sum_i m_i X_i$.
We call $F$ as a {\it tropical polynomial} in $n$ variables $X_1,\ldots,X_n$.
In this manner, min-plus algebra is obtained as the tropicalization of $K$. 
Tropical geometry is algebraic geometry of 
min-plus algebra \cite{IMS-Book}, 
which can be interpreted as the tropicalization of 
the algebraic geometry on $K$.

On the other hand, among integrable systems,
there are some rational maps (on $K$) whose tropicalization
gives non-trivial interesting piecewise linear maps
(on $\R$).
We are interested in the case of $K = \overline{\C((t))}$, the field of  
Puiseux series in $t$,
and we expect a similar story as \S \ref{sec:1-1},
i.e. integrable structure of the piecewise-linear maps 
is described by tropical geometry.

\subsection{Contents}

We are interested in two piecewise-linear systems,
the tropical periodic Toda lattice (trop-pToda), and 
the periodic Box-boll system (pBBS).
These systems are obtained from known integrable rational maps:
the former is the tropicalization of discrete Toda lattice
as the name suggests,
and the latter is the special case of the tropical KdV equation.
It is natural to study them with tropical geometry,
since the integrability of the original rational maps
is described by complex algebraic geometry as the previous example.  
But we would like to emphasize that it is highly non-trivial problem,
because the tropicalization is a very formal limiting procedure
and nothing is ensured about how the original complex structure becomes. 

This article is organized as follows:
we introduce the basic notions of tropical curve theory in \S 2,
including tropical Jacobians for tropical curves,
and tropical theta functions, by following \cite{MikhaZhar06}. 
In \S 3, we review the general solution of the trop-pToda 
based on \cite{IT08,IT09,IT09b}. The general isolevel set 
is isomorphic to the tropical Jacobian, and the corresponding 
solution is written in terms of the tropical theta function.  
In \S 4, we present new results on the evolution equation and 
the spectral curve of the pBBS.
The initial value problem of the pBBS was already  
solved by applying crystal theory 
\cite{KTT06} or combinatorics \cite{MIT08},
where the time evolution is linearized on the high-dimensional real torus,
and the tropical theta functions appear in the solution.
We explore the tropical geometrical aspect of the pBBS via the tropical KdV
equation.
We explicitly give the tropical spectral curve of the pBBS,
and show that the above real torus is really the Jacobian of 
the tropical curve.

\section*{Acknowledgments}

The authors thank the organizers of the workshop
``Tropical Geometry and Integrable Systems" 
at University of Glasgow in July 2011.

\section{Tropical curve theory}
\subsection{Tropical curves}

In this article we consider affine tropical curves in $\R^2$
given by tropical polynomials of two variables as
$$
F(X,Y) 
=
\min_{i \in I} [C_i + n_i X + m_i Y]
\qquad C_i \in \R, \, n_i, m_i \in \Z_{\geq 0},
$$
where $I$ is a finite set.
The tropical curve $\Gamma$ given by a tropical polynomial 
$F(X,Y)$ is defined by
$$
\Gamma = \{(X,Y) \in \R^2 ~|~ 
F(X,Y) \text{ is {\it indifferentiable}} \}.
$$
The meaning of ``$(X,Y)$ is indifferentiable'' is that 
$F(X,Y)$ is accomplished by more than one terms in $F(X,Y)$
at $(X,Y)$.

\begin{Example}
See Figure \ref{i:fig:trop-examples} for examples of tropical curves.
(i) is given by $F(X,Y) = \min(X, \ Y, \ 1)$,
and (ii) is given by $F(X,Y) = \min[2Y, Y+2X, Y+X, Y+3, 10]$.
Let us explain the notion of ``indifferentiable'' in (i).
Let $A_1$, $A_2$ and $A_3$ be three open domains divided 
by the three boundaries $l_{12}$, $l_{23}$ and $l_{13}$,
and let $P$ be the intersection point $l_{12} \cap l_{23} \cap l_{13}$. 
The function $F(X,Y) = \min(X,\ Y,\ 1)$ is ``differentiable'' at 
$(X,Y) \in A_1 \cup A_2 \cup A_3$, 
since we have $F(X,Y) = 1$ in $A_1$,  $F(X,Y) = Y$ in $A_2$ and 
$F(X,Y) = X$ in $A_3$. 
On the other hand, $F(X,Y)$ is ``indifferentiable''
at $(X,Y) \in l_{12} \cup l_{23} \cup l_{13}$,
since at least two of $X$, $Y$ and $1$ become the minimum. 
For instance, 
$F(X,Y) = Y = 1$ on $l_{12} \setminus \{P\}$, 
and $F(X,Y) = X = Y = 1$ at $P$.
\end{Example}

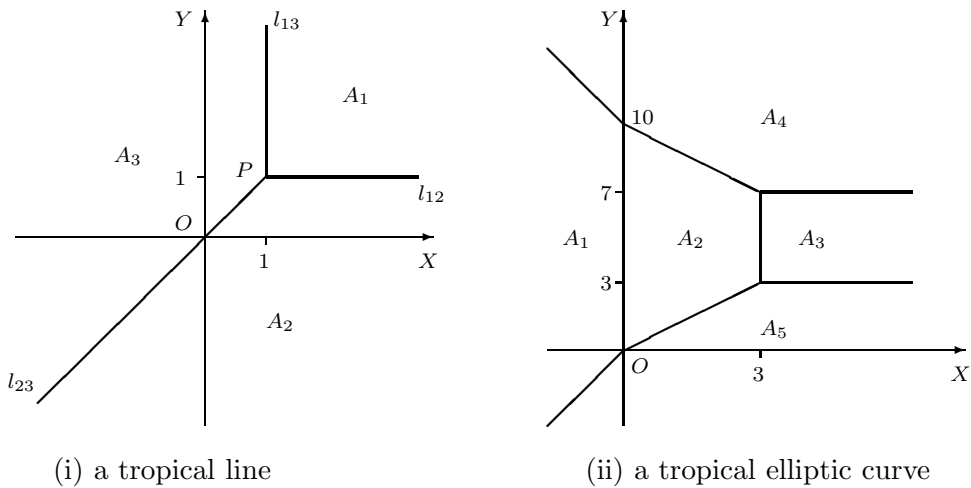
\begin{figure}
\unitlength=1.0mm

\begin{picture}(120,65)(10,-5)

\put(5,30){\vector(1,0){55}}
\put(58,26){\scriptsize$X$}
\put(30,5){\vector(0,1){55}}
\put(26,58){\scriptsize$Y$}
\put(26,31){\scriptsize$O$}

\put(38,30){\line(0,-1){1}}
\put(37,26){\scriptsize$1$}
\put(30,38){\line(-1,0){1}}
\put(26,37){\scriptsize$1$}

\thicklines

\put(8,8){\line(1,1){30}}
\put(38,38){\line(1,0){20}}
\put(38,38){\line(0,1){20}}

\put(48,48){\scriptsize$A_1$}
\put(38,18){\scriptsize$A_2$}
\put(18,40){\scriptsize$A_3$}
\put(34,38){\scriptsize$P$}
\put(39,58){\scriptsize$l_{13}$}
\put(58,35){\scriptsize$l_{12}$}
\put(4,10){\scriptsize$l_{23}$}

\thinlines
\put(75,15){\vector(1,0){55}}
\put(128,11){\scriptsize$X$}
\put(85,5){\vector(0,1){55}}
\put(82,58){\scriptsize$Y$}
\put(86,12){\scriptsize$O$}

\put(86,45){\scriptsize$10$}
\put(85,36){\line(-1,0){1}}
\put(82,35){\scriptsize$7$}
\put(85,24){\line(-1,0){1}}
\put(82,23){\scriptsize$3$}
\put(103,15){\line(0,-1){1}}
\put(102,11){\scriptsize$3$}

\thicklines

\put(85,15){\line(2,1){18}}
\put(103,24){\line(1,0){20}}
\put(103,24){\line(0,1){12}}
\put(85,45){\line(2,-1){18}}
\put(103,36){\line(1,0){20}}
\put(85,45){\line(-1,1){10}}
\put(85,15){\line(-1,-1){10}}
\put(85,15){\line(0,1){30}}

\put(77,29){\scriptsize$A_1$}
\put(92,29){\scriptsize$A_2$}
\put(108,29){\scriptsize$A_3$}
\put(103,45){\scriptsize$A_4$}
\put(103,17){\scriptsize$A_5$}

\put(10,-2){(i) a tropical line}
\put(80,-2){(ii) a tropical elliptic curve}
\end{picture}

\caption{Tropical curves}
\label{i:fig:trop-examples}
\end{figure}

The edges in tropical curves have rational slopes,
and we associate each vertex with a {\it primitive tangent vector}
which is a tangent vector given by a pair of coprime integers.
The primitive tangent vector is uniquely determined up to sign.
(If one of the integers is zero, then let another be $\pm 1$.)

\begin{Definition}\label{i:def:smoothness}
The tropical curve $\Gamma \subset \R^2$ is smooth if
the following two conditions hold:
\begin{itemize}
\item[(a)]
all vertices in $\Gamma$ are $3$-valent. 
\item[(b)]
For each $3$-valent vertex $v$, let $\xi_1, \xi_2, \xi_3$ be the
primitive tangent vectors which are outgoing from $v$.
Then these vectors satisfy
$\xi_1+\xi_2+\xi_3 = (0,0)$ and $|\xi_i \wedge \xi_j| = 1$ for
$i,j \in \{1,2,3\}, ~i \neq j$.  
(For two vectors $\xi = (n_1,n_2)$ and $\xi' = (n_1',n_2')$,
we set $\xi \wedge \xi'=n_1 n_2' - n_2 n_1'$.) 
\end{itemize}
When a tropical curve $\Gamma$ is smooth, 
the {\it genus} of $\Gamma$ is dim$ \ H_1(\Gamma, \Z)$.
\end{Definition}

The two tropical curves at Figure \ref{i:fig:trop-examples} are 
smooth, and the genera are zero and one respectively.

A smooth tropical curve is equipped with 
the {\it metric structure} as follows
(We omit the metric structure for non-smooth tropical curves
for simplicity. See \cite{MikhaZhar06}.):
\begin{Definition}
Assume $\Gamma$ is a smooth tropical curve.
Let $E(\Gamma)$ be the set of edges in $\Gamma$,
and let $\xi_e$ be the primitive tangent vector of $e \in E(\Gamma)$.
We define the {\it length} of edges $l: ~E(\Gamma) \to \R_{\geq 0}$ by
$$
  e \mapsto l(e) = \frac{\parallel e \parallel}
  {\parallel \xi_e \parallel},
$$
where $\parallel~~\parallel$ is any norm in $\R^2$.
\end{Definition}

\subsection{Abelian integral and tropical Jacobian.}

Let $\Gamma$ be a smooth tropical curve whose genus $g$ is not zero.
We fix $g$ generators $B_1,\cdots,B_g$ of the fundamental group of 
$\Gamma$. 
We define the {\it bilinear form} of two paths $p$ and $p'$ on $\Gamma$ by
$$
\langle p ,~p' \rangle
=
\text{``the oriented length of the common part of $p$ and $p'$''}.  
$$
Here ``oriented'' means the sign 
depending on the directions of the two paths on the common part.

\begin{Example}
See Figure \ref{i:fig:trop-example2} for 
the smooth tropical curve $\Gamma$ given by
\begin{align*}
F(X,Y) = \min(2Y, \ Y+3 X, \ Y+2X, \ Y+X+1, \ Y+4, 11).
\end{align*}
The genus of $\Gamma$ is $2$, and we fix the basis $B_1$ and $B_2$ 
of the fundamental group of $\Gamma$ as depicted.
The bilinear forms of $B_1$ and $B_2$ take the values as
$$
\langle B_1, B_1 \rangle = 20, 
\qquad 
\langle B_1, B_2 \rangle = -7,
\qquad 
\langle B_2, B_2 \rangle = 14.
$$
Let us demonstrate how to compute $\langle B_1, B_2 \rangle$:
the common part of $B_1$ and $B_2$ is the edge $PQ$, and we have $l(PQ) = 7$.
Moreover, the directions of $B_1$ and $B_2$ on $PQ$ are opposite,
and we obtain $-l(PQ) = -7$ as the oriented length of $B_1 \cap B_2$.
\end{Example}

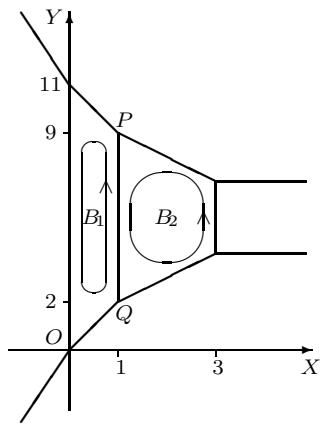
\begin{figure}
\unitlength=0.8mm

\begin{picture}(100,70)(-20,0)

\put(10,10){\vector(1,0){50}}
\put(58,6){\scriptsize$X$}
\put(20,0){\vector(0,1){66}}
\put(16,64){\scriptsize$Y$}

\thicklines

\put(20,10){\line(-2,-3){8}}
\put(20,10){\line(1,1){8}}
\put(28,18){\line(2,1){16}}
\put(44,26){\line(1,0){15}}
\put(20,54){\line(-2,3){8}}
\put(20,54){\line(1,-1){8}}
\put(28,46){\line(2,-1){16}}
\put(44,38){\line(1,0){15}}
\put(20,10){\line(0,1){44}}
\put(28,18){\line(0,1){28}}
\put(44,26){\line(0,1){12}}

\thinlines

\put(24,32){\oval(4,25)}
\put(24.8,36){\scriptsize$\wedge$}
\put(22,31){\scriptsize${B_{\!1}}$}
\put(36,32){\oval(12,15)}
\put(40.8,31){\scriptsize$\wedge$}
\put(34,31){\scriptsize$B_{\!2}$}

\put(15,53){\scriptsize$11$}
\put(16,11){\scriptsize$O$}
\put(28,10){\line(0,-1){1}}
\put(27.5,6){\scriptsize$1$}
\put(44,10){\line(0,-1){1}}
\put(43.5,6){\scriptsize$3$}
\put(20,18){\line(-1,0){1}}
\put(16,17){\scriptsize$2$}
\put(20,46){\line(-1,0){1}}
\put(16,45){\scriptsize$9$}

\put(27.5,47){\scriptsize$P$}
\put(27.5,15){\scriptsize$Q$}
\end{picture}
\caption{Tropical curve of genus $2$}
\label{i:fig:trop-example2}
\end{figure}

Now we introduce the {\it abelian integral} and the {\it tropical 
Jacobian} of $\Gamma$:
\begin{Definition}
Fix $P_0 \in \Gamma$.
The {\it abelian integral} $\psi: ~\Gamma \to \R^g$ is given by
$$
P \mapsto \psi(P) 
= (\langle B_i, \overset{\curvearrowright}{P_0 P} \rangle)_{i=1,\ldots,g},
$$
where $\overset{\curvearrowright}{P_0 P}$ is a path from 
$P_0$ to $P$.
The map $\psi$ induces the map from a set of divisors $\mathrm{Div}(\Gamma)$
on $\Gamma$ to $\R^g;$
$$
  \sum_{i \in I} n_i P_i \mapsto \sum_{i \in I} n_i \, \psi(P_i),
$$
where $I$ is a finite set and $n_i \in \Z$.
\end{Definition}

\begin{Definition}\label{i:def:omega}
The period matrix $\Omega$ of $\Gamma$ is given by
\begin{align}\label{i:eq:tropOmega}
  \Omega = (\langle B_i, B_j\rangle)_{i,j=1,\ldots,g} \in \mathrm{Mat}(g;\R).
\end{align}
The tropical Jacobian variety $J(\Gamma)$ of $\Gamma$ is the
$g$-dimensional real torus given by
\begin{align}\label{i:eq:tropJ}
  J(\Gamma) = \R^g / \Omega \Z^g.
\end{align}
\end{Definition}

\begin{Example}
The tropical curve of genus $1$ depicted at
Figure \ref{i:fig:trop-examples} (ii)
has the period matrix $\Omega = 20$,
and the Jacobian is $\R / 20 \Z$.
As for the tropical curve of genus $2$ depicted at Figure
\ref{i:fig:trop-example2},
the period matrix and the Jacobian are as  
$$
\Omega = \begin{pmatrix} 20 &-7 \\ -7 & 14 \end{pmatrix},
\qquad 
J(\Gamma) = \R^2 / \Omega \Z^2.
$$
\end{Example}

\begin{Remark}
The matrix $\Omega$ is symmetric and positive definite by definition,
and $J(\Gamma)$ is a tropical analogue of Jacobian variety.
By removing all infinite edges of $\Gamma$,
we obtain the maximal compact subgraph $\Gamma^{\mathrm{cpt}}$  
of $\Gamma$.
The map $\psi$ is not injective since 
$\overset{\curvearrowright}{P_0 P}$ is not unique,
but the induced map $\Gamma^{\mathrm{cpt}} \to J(\Gamma)$ becomes injective.
When $g=1$,
$\psi$ induces $\Gamma^{\mathrm{cpt}} \stackrel{\sim}{\to} J(\Gamma)$.
\end{Remark}

\begin{Remark}
There is a well-defined notion of rational equivalence class 
in $\Div(\Gamma)$. Let $\Pic^k(\Gamma)$ be the rational equivalence class of 
$\Div^k(\Gamma)$, where $\Div^k(\Gamma) \subset \Div(\Gamma)$ is 
a set of divisors of degree $k$. Then we have a commutative diagram: 
$$
\begin{matrix}
\Div^k(\Gamma) & \to & J(\Gamma)
\\
& \searrow & \uparrow_{\stackrel{\beta}{}} 
\\[1mm]
&  & \Pic^k(\Gamma)  
\end{matrix}
$$
where the map $\beta$ is an isomorphism \cite{MikhaZhar06}.
\end{Remark}

\subsection{Tropical theta function.}

Fix a positive integer $g$ and a symmetric and positive definite 
matrix $\Omega \in \mathrm{Mat}(g;\R)$.
(Here the matrix $\Omega$ is not always 
a period matrix of some tropical curve.) 

\begin{Definition}\label{i:def:theta}
The {\it tropical theta function} $\Theta({\bf Z}; \Omega)$ 
of ${\bf Z}\in \R^g$ is defined by 
$$
\Theta({\bf Z};\Omega)
=
\min_{{\bf n} \in \Z^g} \Bigl\{{\bf n} \cdot 
\Bigl(\frac{1}{2} \Omega {\bf n} + {\bf Z}\Bigr) \Bigr\}.
$$   
We call the $g$-dimensional real torus given by
\begin{align}\label{i:eq:g-J}
J_\Omega = \R^g / \Omega \Z^g
\end{align} 
the {\it principally polarized tropical abelian variety}.
(If $\Omega$ is the period matrix of a tropical curve $\Gamma$,
then $J_\Omega$ is nothing but the tropical Jacobian $J(\Gamma)$.)
\end{Definition}

It is easy to see the following:
\begin{Lemma}
The function $\Theta({\bf Z}) = \Theta({\bf Z}; \Omega)$ 
satisfies the quasi-periodicity:
\begin{align}\label{i:theta-quasi}
\Theta({\bf Z} + \Omega {\bf m})
=
- {\bf m} \cdot \Bigl(\frac{1}{2} \Omega {\bf m} + {\bf Z}\Bigr) 
+ \Theta({\bf Z})
\qquad {\bf m} \in \Z^g.
\end{align}
\end{Lemma}

\begin{Remark}
Recall the Riemann's theta function:
\begin{align}
\theta({\bf z}; W) 
= 
\sum_{{\bf n} \in \Z^g} \exp \bigl(\pi \sqrt{-1} \ {\bf n} \cdot 
(W {\bf n} +2 {\bf z})\bigr)
\qquad {\bf z} \in \mathbb{C}^g,
\end{align}
where $W \in \mathrm{Mat}(g;\mathbb{C})$ is symmetric and ${\rm Im} W$ 
is positive definite.
This function satisfies the periodicity and 
quasi-periodicity: 
\begin{align}
\begin{split}
&\theta({\bf z} + {\bf m}; W) = \theta({\bf z}; W),
\\
&\theta({\bf z}+ K {\bf m}; W) 
=  
\exp \bigl(-\pi \sqrt{-1} \ {\bf m} \cdot (W {\bf m} +2 {\bf z})\bigr)
\theta({\bf z}; W),
\end{split}
\end{align}
for ${\bf m} \in \Z^g$.
Remark that only the quasi-periodicity remains in the tropical case.
\end{Remark}

\section{Tropical periodic Toda lattice}

\subsection{Introduction}

The tropical periodic Toda lattice (trop-pToda) is given by 
the piecewise-linear evolution equation: 
\begin{align}\label{i:eq:UDpToda}
\begin{split}
&Q_j^{t+1} = \min(W_j^t, Q_j^t-X_j^t),
\qquad X_j^t = \min_{0 \leq k \leq N-1}
\bigl(\sum_{l=1}^k (W_{j-l}^t - Q_{j-l}^t)\bigr),
\\
&W_j^{t+1} = Q_{j+1}^t+W_j^t - Q_j^{t+1}
\end{split}
\end{align}
on the phase space
$$
\displaystyle{
\mathcal{T} 
= 
\{(Q_j,W_j)_{j \in \Z_N} ~|~ \sum_{j=1}^N Q_j < \sum_{j=1}^N W_j \}
\subset \R^{2N}.}
$$
(In \cite{IT08}, this system is called the ultradiscrete periodic 
Toda lattice,
where ``ultradiscrete'' means ``tropical'' in our present terminology.)

\begin{Remark}
The trop-pToda is the tropicalization of 
the discrete $N$-periodic Toda lattice \cite{HT95} given by
$$
q_j^{t+1} = q_j^t + w_j^t - w_{j-1}^{t+1},
\qquad
w_j^{t+1} = \frac{q_{j+1}^t w_j^t}{q_j^{t+1}}
$$
on the phase space $\{(q_j,w_j)_{j \in \Z_N} \} \simeq K^{2N}$ 
under the setting:
$$
\sum_{j=1}^N \left(\val(w_j^t)-\val(q_j^t)\right) > 0,
\qquad
Q_j^t = \val(q_j^t),
\qquad 
W_j^t = \val(w_j^t).
$$
See \cite[Prop. 2.1]{KimijimaTokihiro02} for the detail,
where the strategy is essentially same as Lemma \ref{lem:P}.

The Toda lattice equation \eqref{originalToda}
is a continuous limit $\delta \to 0$ of the above discrete Toda lattice, 
with $w_j^t = \delta^2 b_j$
and $q_j^t = 1 + \delta a_j'$.
Here $\delta$ is a unit of the discrete time.
\end{Remark}

The system \eqref{i:eq:UDpToda} has $N+1$ conserved tropical polynomials
$H_k ~(k=1,\ldots,N+1)$
on $\mathcal{T}$. Here are some of them:
\begin{align}\label{i:udpToda-H}
\begin{split}
&H_1 = \min_{1 \leq j \leq N}(Q_j,W_j),
\\  
&H_2 = \min\bigl(\min_{1\leq i<j \leq N} (Q_i+Q_j),
            \min_{1\leq i<j \leq N} (W_i+W_j),
            \min_{1\leq i,j\leq N, j \neq i,i-1} (Q_i+W_j)\bigr),
\\
&H_N = \min \bigl(\sum_{j=1}^N Q_j, \sum_{j=1}^N W_j \bigr),
\\ 
&H_{N+1} = \sum_{j=1}^N (Q_j + W_j).
\end{split}
\end{align}

\subsection{General solution}

Fix $C = (C_k)_{k=1,\cdots,N+1} \in \R^{N+1}$, and define the 
isolevel set $\mathcal{T}_C$ by
\begin{align}\label{i:TC}
\mathcal{T}_C 
= 
\{ \tau \in \mathcal{T} ~|~ H_k(\tau) = C_k ~(k=1,\cdots,N+1) \}.
\end{align}
We are to describe the general solution to \eqref{i:eq:UDpToda} and 
the isolevel set $\mathcal{T}_C$ in terms of tropical geometry.

Let $\Gamma_C$ be the tropical curve given by tropical polynomial 
\begin{align}\label{i:trop-Toda}
F(X,Y)
= 
\min\bigl(2Y, Y+ \min(NX, (N-1)X+C_1, \ldots, X+C_{N-1}, C_N), C_{N+1}\bigr).
\end{align}
We call $\Gamma_C$ the spectral curve of the trop-pToda.
Remark that $F(X,Y)$ corresponds to the tropicalization of
the l.h.s of \eqref{Toda-curve} (the defining equation for the algebraic 
curve $\gamma_c$.)

We set $L$, $\lambda_k$ and  $p_k$ for $k = 0,\ldots,N-1$ as 
\begin{align}\label{i:C-lambda}
\begin{split}
&L = C_{N+1} - 2 (N-1) C_1,
\\
&\lambda_0 = 0, \qquad \lambda_k = C_{k+1} - C_k
\qquad k=1,\ldots,N-1, 
\\
&p_0=L, \qquad p_k = L - 2 \sum_{j=1}^{N-1} \min(\lambda_k,\lambda_j) 
\qquad k=1,\ldots,N-1.
\end{split}
\end{align}
The curve $\Gamma_C$ is smooth if and only if 
$\lambda_1 < \lambda_2 < \cdots < \lambda_{N-1}$ and $p_k > 0$ for 
$k=1,\ldots,N-1$. Assume $\Gamma_C$ is smooth, 
then the genus $g$ is $g=N-1$. 
See Figure 3 for $\Gamma_C$, where we set $C_1 =0$ for simplicity.

Fix the basis $B_1,\ldots,B_g$ of the fundamental group $\pi_1(\Gamma_C)$
as Figure 3.
The period matrix $\Omega$ \eqref{i:eq:tropOmega} of $\Gamma_C$ is 
obtained as
\begin{align}\label{i:omega}
\Omega_{ij} = 
\begin{cases}
p_{i-1} + p_{i} + 2 (\lambda_{i} - \lambda_{i-1}) 
\qquad i=j 
\\
-p_i \qquad j=i+1
\\
-p_j \qquad i=j+1
\\
0 \qquad \text{otherwise}
\end{cases}
\end{align}
and we get the tropical Jacobian of $\Gamma_C$ as 
\begin{align}\label{pToda-J}
J(\Gamma_C) = \R^g / \Omega \Z^g.
\end{align}

\begin{figure}\label{i:fig:curveToda}
\unitlength=0.8mm
\begin{picture}(140,78)(-25,3)

\put(20,10){\vector(1,0){70}}
\put(0,10){\line(1,0){20}}
\put(91,9){\scriptsize $X$}

\thicklines

\put(15,74){\scriptsize $L$}
\put(20,6){\scriptsize $0$}
\put(28,10){\line(0,-1){1}}
\put(27,6){\scriptsize $\lambda_1$}
\put(36,10){\line(0,-1){1}}
\put(35,6){\scriptsize $\lambda_2$}
\put(45,6){$\cdots$}
\put(56,10){\line(0,-1){1}}
\put(55,6){\scriptsize $\lambda_{g-1}$}
\put(68,10){\line(0,-1){1}}
\put(67,6){\scriptsize $\lambda_{g}$}

\put(20,10){\line(-1,-2){4}}
\put(20,10){\line(4,5){8}}
\put(28,20){\line(1,1){8}}
\put(36,28){\line(2,1){10}}
\put(53,35){\line(5,1){15}}
\put(68,38){\line(1,0){15}}

\put(20,75){\line(-1,2){4}}
\put(20,75){\line(4,-5){8}}
\put(28,65){\line(1,-1){8}}
\put(36,57){\line(2,-1){10}}
\put(53,50){\line(5,-1){15}}
\put(68,47){\line(1,0){15}}

\put(20,10){\line(0,1){65}}
\put(28,20){\line(0,1){45}}
\put(36,28){\line(0,1){29}}
\put(56,35.5){\line(0,1){14}}
\put(68,38){\line(0,1){9}}
\put(46,42){$\cdots$}

\thinlines 

\put(24,42){\oval(4,45)}
\put(26,32){\vector(0,1){3}}
\put(22,29){\scriptsize $B_{\!1}$}
\put(32,42){\oval(4,27)}
\put(34,37){\vector(0,1){3}}
\put(30,34){\scriptsize $B_{\!2}$}
\put(62,42){\oval(8,8)}
\put(66,41){\vector(0,1){3}}
\put(61,41){\scriptsize $B_{\!g}$}

\end{picture}
\caption{Spectral curve for the trop-pToda}
\end{figure}
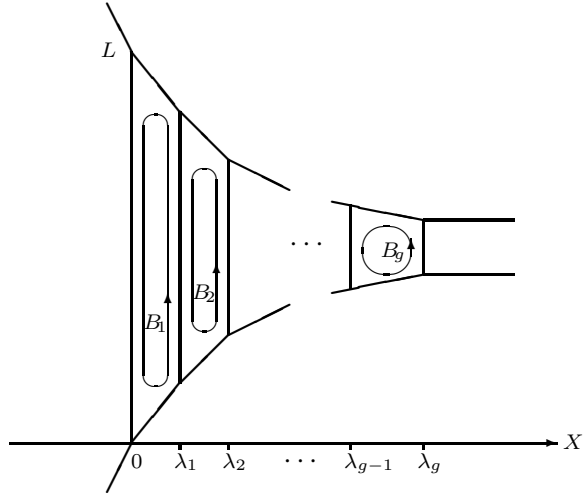

\begin{Theorem}\label{i:th:pToda} 
When $\Gamma_C$ is smooth, we have the following:
\\
(i) {\rm \cite[Th.~3.5]{IT09}} 
Fix ${\bf Z}_0 \in \R^g$ and define 
$T_n^t = \Theta({\bf Z}_0 + \boldsymbol{\lambda} t - L {\bf e}_1 n;\Omega)$,
where 
$$
\boldsymbol{\lambda} = (\lambda_1, \lambda_2 -\lambda_1, \ldots,
\lambda_g-\lambda_{g-1}), ~
{\bf e}_1 = (1,0,\ldots,0) \in \R^g.
$$ 
The solution for the trop-pToda is given by
\begin{align}\label{i:QWsol}
\begin{split}
  &Q_n^t = T_{n-1}^{t} + T_{n}^{t+1}
           - T_{n-1}^{t+1} - T_{n}^{t}+C_1,
  \\
  &W_n^t = T_{n-1}^{t+1} + T_{n+1}^t-
          T_{n}^{t} - T_{n}^{t+1} + L+C_1.
\end{split}
\end{align}
(ii) {\rm \cite[Th.~1.3]{IT09b}} This solution induces the isomorphism 
$J(\Gamma_C) \stackrel{\sim}{\rightarrow} \mathcal{T}_C$.
In particular, the time evolution of the trop-pToda is linearized
on $J(\Gamma_C)$, whose velocity is $\boldsymbol{\lambda}$.
\end{Theorem}

One of the keys to prove Theorem \ref{i:th:pToda} (i) 
is the following lemma:

\begin{Lemma}\cite[Prop. 2.10]{IT09}
The function $T_n^t$ satisfies
$$
T_n^{t+1} + T_n^{t-1} 
= 
\min [ 2 T_n^t, ~ T_{n-1}^{t+1} + T_{n+1}^{t-1} + L ].
$$
\end{Lemma}

This identity corresponds to a tropicalization of the bilinear form.
This lemma is proved by applying the tropical version of 
Fay's trisecant identity for tropical theta functions.
See \cite{IT09} for the detail.

\begin{Example} 
The case of $N=2$.
The curve $\Gamma_C$ is smooth if and only if $C_3 > 2 C_2 > 4 C_1$.
In this simplest case, we can explicitly construct the isomorphism
$\alpha$:
\begin{align*}
\begin{matrix}
\mathcal{T}_C & \stackrel{\alpha}{\to} & \Gamma_C^{\mathrm{cpt}}
& \stackrel{\psi}{\to} & J(\Gamma_C) \\
(Q_1,W_1,Q_2,W_2) & \mapsto & P=(\min(Q_2,W_1), Q_1+W_1) & \mapsto 
& \langle B_1,\overset{\curvearrowright}{P_0 P} \rangle
\end{matrix}.
\end{align*}
The solution \eqref{i:QWsol} induces the inverse map of
$\psi \circ \alpha$. 
Let us consider the case of $C=(0,3,8)$, where $\Gamma_C$ is depicted as 

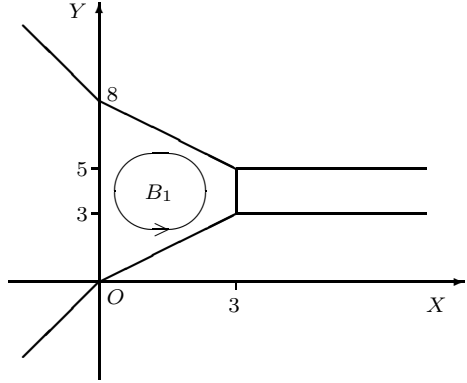
\begin{center}
\unitlength=1.0mm
\begin{picture}(60,55)(10,-12)

\put(3,5){\vector(1,0){60}}
\put(58,1){\scriptsize$X$}
\put(15,-8){\vector(0,1){50}}
\put(11,40){\scriptsize$Y$}
\put(16,2){\scriptsize$O$}

\put(16,29){\scriptsize$8$}
\put(15,20){\line(-1,0){1}}
\put(12,19){\scriptsize$5$}
\put(15,14){\line(-1,0){1}}
\put(12,13){\scriptsize$3$}
\put(33,5){\line(0,-1){1}}
\put(32,1){\scriptsize$3$}

\put(23,17){\oval(12,10)}
\put(22,11.2){\scriptsize$>$}
\put(21,16){\scriptsize${B_1}$}

\thicklines

\put(15,5){\line(2,1){18}}
\put(33,14){\line(1,0){25}}
\put(33,14){\line(0,1){6}}
\put(15,29){\line(2,-1){18}}
\put(33,20){\line(1,0){25}}
\put(15,29){\line(-1,1){10}}
\put(15,5){\line(-1,-1){10}}
\put(15,5){\line(0,1){24}}

\end{picture}
\end{center}

\noindent
The following is an example of linearization, 
where one sees $\boldsymbol{\lambda} = (3)$.
We set $P_0 = O$:
\begin{align*}
\begin{matrix}
& & \mathcal{T}_C = \{(Q_1,W_1,Q_2,W_2)\} & \stackrel{\alpha}{\to} & 
\Gamma_C^{\mathrm{cpt}} & \stackrel{\psi}{\to} & J(\Gamma_C) \simeq \R/16 \Z 
\\[1mm]
\texttt{t=0}  & & (3,4,0,1) & & (0,7) & & 9
\\
\texttt{t=1} & & (3,1,0,4) & & (0,4) & & 12
\\
\texttt{t=2} & & (1,0,2,5) & & (0,1) & & 15
\\
\texttt{t=3} & & (0,2,3,3) & & (2,2) & & 2 \equiv 18
\\
\texttt{t=4} & & (0,5,3,0) & & (3,5) & & 5 \equiv 21
\end{matrix}
\end{align*}
\end{Example}

For general $N>2$, the isomorphism 
$\mathcal{T}_C \stackrel{\sim}{\to} J(\Gamma_C)$
is regarded as a composition of the injective map
$\alpha: \mathcal{T}_C \to \mathrm{Div}_{\mathrm{eff}}^{g}(\Gamma_C)$
and the abelian integral $\psi$,
but $\alpha$ becomes too complicated.

\section{Periodic BBS}

\subsection{Introduction}

The periodic Box-ball system (pBBS) is a cellular automaton 
defined by adding a periodic boundary condition \cite{YT02}
to the original (infinite) Box-ball system \cite{TS90}.
Let $L$ be the number of boxes aligned on an oriented circle.
Put $M < L/2$ balls into the boxes,
assuming that each box can accommodate one ball at most.
Move the balls with the following rule
which defines the time evolution from $t$ to $t+1$:
\begin{enumerate}
\item 
Connect ``an occupied box whose immediate right is empty" 
and the empty box with an arc.
Do the same for all such boxes. 
\item
In the rest, do the same as (1) by ignoring the connected boxes and arcs.
\item
Continue the same procedure as (2) until all occupied boxes are connected
with empty boxes.
\item
Move all balls to the connected empty boxes.
\end{enumerate}
In the process, the term ``right" is understood 
along the direction of the orientation of the circle. 
This evolution is determined uniquely and invertible.  
Let us show an example of $L=11$ and $M=4$ in the following,
where we identify the left and right boundaries with thick lines.
The above rule works as

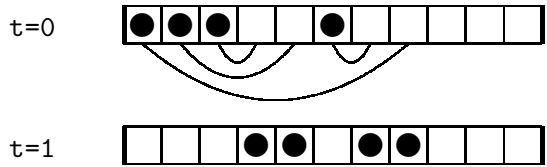
\begin{center}
\unitlength=1mm
\begin{picture}(140,30)(10,30)

\put(35,51){\texttt{t=0}}
\put(50,55){\line(1,0){55}}
\put(50,50){\line(1,0){55}}
\multiput(50,50)(5,0){12}{\line(0,1){5}}
\thicklines
\multiput(50,50)(0.2,0){2}{\line(0,1){5}}
\multiput(105,50)(0.2,0){2}{\line(0,1){5}}

\thinlines

\put(52.5,52.5){\circle*{3}}
\put(57.5,52.5){\circle*{3}}
\put(62.5,52.5){\circle*{3}}

\put(77.5,52.5){\circle*{3}}


\qbezier(62.5,50)(65,45)(67.5,50)
\qbezier(77.5,50)(80,45)(82.5,50)

\qbezier(57.5,50)(65,41)(72.5,50)

\qbezier(52.5,50)(70,35)(87.5,50)


\put(35,35){\texttt{t=1}}
\put(50,39){\line(1,0){55}}
\put(50,34){\line(1,0){55}}
\multiput(50,34)(5,0){12}{\line(0,1){5}}
\thicklines
\multiput(50,34)(0.2,0){2}{\line(0,1){5}}
\multiput(105,34)(0.2,0){2}{\line(0,1){5}}

\thinlines
\put(67.5,36.5){\circle*{3}}
\put(72.5,36.5){\circle*{3}}

\put(82.5,36.5){\circle*{3}}
\put(87.5,36.5){\circle*{3}}

\end{picture}
\end{center}
and we obtain more as follows:
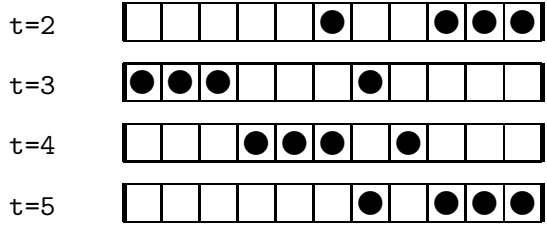
\begin{center}
\unitlength=1.0mm
\begin{picture}(90,35)(0,8)

\thinlines
\put(0,35){\texttt{t=2}}
\put(15,39){\line(1,0){55}}
\put(15,34){\line(1,0){55}}
\multiput(20,34)(5,0){10}{\line(0,1){5}}
\thicklines
\multiput(15,34)(0.2,0){2}{\line(0,1){5}}
\multiput(70,34)(0.2,0){2}{\line(0,1){5}}

\put(42.5,36.5){\circle*{3}} 
\put(57.5,36.5){\circle*{3}} 
\put(62.5,36.5){\circle*{3}} 
\put(67.5,36.5){\circle*{3}} 


\thinlines
\put(0,27){\texttt{t=3}}
\put(15,31){\line(1,0){55}}
\put(15,26){\line(1,0){55}}
\multiput(20,26)(5,0){10}{\line(0,1){5}}
\thicklines
\multiput(15,26)(0.2,0){2}{\line(0,1){5}}
\multiput(70,26)(0.2,0){2}{\line(0,1){5}}

\put(47.5,28.5){\circle*{3}} 
\put(27.5,28.5){\circle*{3}} 
\put(17.5,28.5){\circle*{3}} 
\put(22.5,28.5){\circle*{3}} 


\thinlines
\put(0,19){\texttt{t=4}}
\put(15,23){\line(1,0){55}}
\put(15,18){\line(1,0){55}}
\multiput(20,18)(5,0){10}{\line(0,1){5}}
\thicklines
\multiput(15,18)(0.2,0){2}{\line(0,1){5}}
\multiput(70,18)(0.2,0){2}{\line(0,1){5}}

\put(32.5,20.5){\circle*{3}} 
\put(37.5,20.5){\circle*{3}} 
\put(42.5,20.5){\circle*{3}} 
\put(52.5,20.5){\circle*{3}} 


\thinlines
\put(0,11){\texttt{t=5}}
\put(15,15){\line(1,0){55}}
\put(15,10){\line(1,0){55}}
\multiput(20,10)(5,0){10}{\line(0,1){5}}
\thicklines
\multiput(15,10)(0.2,0){2}{\line(0,1){5}}
\multiput(70,10)(0.2,0){2}{\line(0,1){5}}

\put(47.5,12.5){\circle*{3}} 
\put(57.5,12.5){\circle*{3}} 
\put(62.5,12.5){\circle*{3}} 
\put(67.5,12.5){\circle*{3}} 

\end{picture}
\end{center}
One can observe that the larger series of balls overtakes 
the smaller one repeatedly.
We call a series of balls as a {\it soliton}.
When a state has $g$ solitons, we call the state as a $g$-soliton state.
The system has a finite configuration space with $\binom{L}{M}$ states, 
thus any state comes back to itself in a finite time.
We identify an occupied box and an empty box with $1$ and $0$ respectively, 
and define the phase space of the pBBS by 
\begin{align}
\mathcal{U}_{\rm BBS} = 
\left\{(U_k)_{k \in \Z_L} ~|~ U_k \in \{0,1\}, 
~\sum_{k=1}^L U_k < \frac{L}{2} \right\}.
\end{align}
We write $U_k^t$ for the number of balls in the $k$-th box at time $t$.

The pBBS has the conserved quantities $\mu =(\mu_1,\mu_2,\mu_3,\ldots)$
constructed as the by-product of the evolution rule;
let $\mu_1$ (resp. $\mu_2$) be the number of the arcs drown 
at the step (1) (resp. (2)), and so on.
Let us calculate $\mu$ of the above example,
where the arc with the number $i$ contributes to
$\mu_i$:

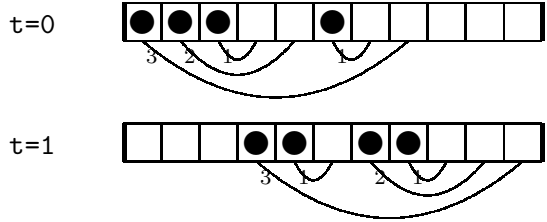
\begin{center}
\unitlength=1mm
\begin{picture}(140,35)(10,25)

\put(35,51){\texttt{t=0}}
\put(50,55){\line(1,0){55}}
\put(50,50){\line(1,0){55}}
\multiput(50,50)(5,0){12}{\line(0,1){5}}
\thicklines
\multiput(50,50)(0.2,0){2}{\line(0,1){5}}
\multiput(105,50)(0.2,0){2}{\line(0,1){5}}

\thinlines

\put(52.5,52.5){\circle*{3}}\put(53,47){\scriptsize$3$} 
\put(57.5,52.5){\circle*{3}}\put(58,47){\scriptsize$2$}  
\put(62.5,52.5){\circle*{3}}\put(63,47){\scriptsize$1$}  

\put(77.5,52.5){\circle*{3}}\put(78,47){\scriptsize$1$}   


\qbezier(62.5,50)(65,45)(67.5,50)
\qbezier(77.5,50)(80,45)(82.5,50)

\qbezier(57.5,50)(65,41)(72.5,50)

\qbezier(52.5,50)(70,35)(87.5,50)


\put(35,35){\texttt{t=1}}
\put(50,39){\line(1,0){55}}
\put(50,34){\line(1,0){55}}
\multiput(50,34)(5,0){12}{\line(0,1){5}}
\thicklines
\multiput(50,34)(0.2,0){2}{\line(0,1){5}}
\multiput(105,34)(0.2,0){2}{\line(0,1){5}}

\thinlines
\put(67.5,36.5){\circle*{3}}\put(68,31){\scriptsize$3$}   
\put(72.5,36.5){\circle*{3}}\put(73,31){\scriptsize$1$}    

\put(82.5,36.5){\circle*{3}}\put(83,31){\scriptsize$2$}     
\put(87.5,36.5){\circle*{3}}\put(88,31){\scriptsize$1$}      


\qbezier(72.5,34)(75,29)(77.5,34)
\qbezier(87.5,34)(90,29)(92.5,34)

\qbezier(82.5,34)(90,25)(97.5,34)

\qbezier(67.5,34)(85,19)(102.5,34)

\end{picture}
\end{center}
We see that $\mu = (2,1,1)$
is invariant under the evolution.

By the definition, $\mu_1$ is the number of the solitons in 
a state.
We introduce an equivalent expression $\lambda$ to $\mu$ by 
$$
\lambda=(\lambda_1, \lambda_2,\ldots, \lambda_{\mu_1}),
\qquad \lambda_j = \#\{\mu_k ~|~ \mu_k \geq \mu_1+1-j \}.
$$
(Do not confuse these $\lambda_j$ with those at \eqref{i:C-lambda}.) 
We have $\lambda_1 \leq \lambda_2 \leq \cdots \leq \lambda_{\mu_1}$ 
and $\sum_{i} \mu_i = \sum_{j=1}^{\mu_1} \lambda_j < \frac{L}{2}$.
In the above example, we get $\lambda=(1,3)$.
We regard $\lambda$ as the map from 
$g$-soliton states in $\mathcal{U}_{\rm BBS}$ to $(\Z_{\geq 0})^g$,
and write $\lambda(U)$ for the image of a 
$g$-soliton state $U \in \mathcal{U}_{\rm BBS}$.

\subsection{General solution.}

We fix $\lambda = (\lambda_1, \lambda_2, \ldots,\lambda_g)$
as $\lambda_1 \leq \lambda_2 \leq \cdots \leq \lambda_g$,
and define the isolevel set: 
$$
\mathcal{U}_{{\rm BBS},\lambda} = \{ U \in \mathcal{U}_{\rm BBS}
 ~|~\lambda(U) = \lambda \}.
$$

The symmetry of the pBBS is described by 
the $\widehat{\mathfrak{sl}}_2$ crystal, and  
the general solution was obtained as follows
(For an introductive review, see \cite{IKT11}.):

\begin{Theorem}\cite{KS06,KTT06} 
\label{th:pBBS}
For $\lambda=(\lambda_1< \cdots <\lambda_g)$, define
$p_k \in \Z_{>0} ~(k=1,\ldots,g)$ and 
$A=(A_{ij})_{i,j=1,\ldots,g} \in \mathrm{Mat}(g;\Z)$ by
\begin{align}\label{A}
p_k = L - 2 \sum_{j=1}^{g} \min(\lambda_k,\lambda_j),
\qquad 
A_{ij} = p_i \delta_{ij} + 2 \min[\lambda_i, \lambda_j].
\end{align}
Then $A$ is symmetric and positive definite.
Define the $g$-dimensional torus $J(L,\lambda)$ as
\begin{align}\label{pBBS-J}
J(L,\lambda) = \R^g / A \Z^g.
\end{align}
Then we have the following:
\\
(i) There is one-to-one correspondence 
between the isolevel set $\mathcal{U}_{{\rm BBS},\lambda}$ and 
the integer points in $J(L,\lambda)$ 
induced by what is called the Kerov-Kirillov-Reshetekhin bijection. 
Let $\Phi$ be the corresponding embedding 
$\Phi: ~\mathcal{U}_{{\rm BBS},\lambda} \to J(L,\lambda)$.
Via $\Phi$, the time evolution of the pBBS is linearized
on $J(L,\lambda)$, whose velocity is $\lambda$.
\\
(ii)
For $U^0 = (U_k^0)_{k \in \Z_L} \in \mathcal{U}_{{\rm BBS},\lambda}$,
set ${\bf Z}_0 = \Phi(U^0) - \frac{{\bf p}}{2}$.
Then $U^t = (U_k^t)_{k \in \Z_L}$ is written in terms of 
tropical theta function $\Theta({\bf Z})=\Theta({\bf Z};A)$ as
\begin{align*}
U_k^t = 
&- \Theta\bigl({\bf Z}_0 - k{\bf v}_1 + t \lambda \bigr)
+ \Theta\bigl({\bf Z}_0 - (k-1){\bf v}_1 + t \lambda \bigr)
\\
&+\Theta\bigl({\bf Z}_0 - k{\bf v}_1 + (t+1) \lambda \bigr)
- \Theta\bigl({\bf Z}_0 -(k-1){\bf v}_1 + (t+1) \lambda \bigr),
\end{align*}
where ${\bf v}_1 = (1,1,\ldots,1) \in \R^g$.
\end{Theorem}

\begin{Remark}
The $g$-dimensional torus \eqref{pBBS-J} is 
the principally polarized tropical abelian variety
(Definition \ref{i:def:theta}).
\end{Remark}

\begin{Remark}\label{rem:Tvm}
There is a family of commutative and invertible time evolutions 
$\{T_m\}_{m \in \Z_{>0}}$ on $\mathcal{U}_{{\rm BBS},\lambda}$ 
\cite[Th. 3.2]{FOY00},
and $T_m$ induces the linear motion on $J(L,\lambda)$
of the velocity ${\bf v}_m = (\min[m,\lambda_j])_{j=1,\ldots,g} \in \Z^g$
\cite[Th. 3.11]{KTT06}.
For $m \geq \lambda_g$, $T_m$ gives the original evolution of the pBBS,
namely, ${\bf v}_m =\lambda$.
\end{Remark}

We will use the following lemma in \S \ref{subsec:4-3}:

\begin{Lemma}\label{lem:U}
Fix $\lambda=(\lambda_1< \cdots <\lambda_g)$.
\\
(i) There is a state $U_0 \in \mathcal{U}_{{\rm BBS},\lambda}$ 
without soliton scattering, 
i.e. the set composed of the lengths of $g$ solitons 
coincides with $\{\lambda_1, \cdots, \lambda_g \}$. 
\\
(ii) For any state $U \in \mathcal{U}_{{\rm BBS},\lambda}$,
there is a sequence of evolutions 
$T := T_{l_1}^{n_1} T_{l_2}^{n_2} \cdots T_{l_k}^{n_k}~
(1 \leq l_1,\ldots,l_k \leq \lambda_g, ~ n_1,\ldots,n_k \in \Z \setminus \{0\})$
such that $T(U) = U_0$. 
\end{Lemma}
\begin{proof}
(i) It follows from the condition $\sum_{k} U_k < \frac{L}{2}$
of $\mathcal{U}_{\rm BBS}$.
(ii) Due to Theorem \ref{th:pBBS} and Remark \ref{rem:Tvm}, 
it is enough to prove that 
$\oplus_{1 \leq m \leq \lambda_g} \Z {\bf v}_m$
includes the basis
$\{{\bf e}_i = (\underbrace{0,\ldots,0}_{i-1},1,\underbrace{0,\ldots,0}_{g-i})
~|~ i=1,\ldots,g \}$ of $\Z^g$.
We actually have 
$$
\sum_{k=i}^g {\bf e}_{k} = -{\bf v}_{\lambda_{i-1}} + {\bf v}_{\lambda_{i-1}+1}
\qquad i=1,\ldots,g,
$$
where we assume $\lambda_0 = 0$ and ${\bf v}_0 = (0,\cdots,0) \in \Z^g$.
Hence the claim follows.
\end{proof}

In the following sections, we are to clarify the tropical 
geometrical origin of $J(L,\lambda)$.

\subsection{Tropical periodic KdV equation and pBBS}

The discrete KdV equation is given by the evolution equation 
\cite{H77}:
\begin{align}\label{KdV}
u_k^{t+1} + \frac{\delta}{u_{k-1}^{t+1}} 
= 
u_{k-1}^{t} + \frac{\delta}{u_{k}^{t}} \quad (k,t \in \Z),
\end{align}
where $\delta$ is a constant element.
We assume the periodic boundary condition $u_{k+L}^t\equiv u_k^t$ for each $k,t$.
Naturally, we regard the index $k$ as an element of $\Z_L$.

The discrete KdV equation (\ref{KdV}) is equivalent to the matrix equation: 
\begin{align}\label{eq4.5}
R^{t+1}S^{t+1}=S^tR^t,
\end{align}
where $S^t=S^t(y)$ and $R^t=R^t(y)$ are elements of $\mathrm{Mat}(L;K[y])$ which are defined by 
\[
S^t=\left(\begin{array}{@{\,}ccccc@{\,}}
	u_1^t & 1 &  &  &  \\
	 & u_2^t & 1 &  &  \\
	 &  & \ddots & \ddots &  \\
	 &  &  &  u_{L-1}^t& 1 \\
	(-1)^{L-1}y &  &  &  & u_L^t
\end{array}\right),\quad
R^t=\left(\begin{array}{@{\,}ccccc@{\,}}
	\frac{\delta}{u_1^t} & 1 &  &  &  \\
	 & \frac{\delta}{u_2^t} & 1 &  &  \\
	 &  & \ddots & \ddots &  \\
	 &  &  &  \frac{\delta}{u_{L-1}^t}& 1 \\
	(-1)^{L-1}y &  &  &  & \frac{\delta}{u_L^t}
\end{array}\right).
\]

Put $X^t(y):=R^t(y)S^t(y)$. Then the equation (\ref{eq4.5}) is rewritten as 
\begin{align}
X^{t+1}S^t=S^tX^t,
\end{align}
which implies that the characteristic polynomial $f(x,y)=\det(X^t(y)+x\cdot\mathrm{id.})$
is invariant under the time evolution.
($-x$ is the eigenvalue of $X(y)$).
The algebraic curve defined by the polynomial $f$ 
is called the \textit{spectral curve} of periodic discrete KdV.
The following lemma follows from \cite[Th. II.1]{MIT05}.

\begin{Lemma}\label{lem:KdV-curve}
Let $M$ be the maximum integer satisfying $M<\frac{L}{2}$.
The polynomial $f$ is written as:
\[
f(x,y)=y^2+y(c_Mx^M+\cdots+c_1x+c_0)+(x+\delta)^L,
\]
where $c_i$ $(i=0,1,\dots,M)$ are rational functions in 
$\R_{>0}(u_n^t,\delta)$ or $\R_{<0}(u_n^t,\delta)$.
\end{Lemma}

The pBBS is obtained as a tropical KdV equation with the periodic boundary condition.
Let $u_k^t,\delta\in K=\overline{\C((t))}$ and
$\val(u_k^t) = U_k^t$, $\val(\delta) = 1$.
Let $\mathcal{U}$ be the subset of $\Q^L$ defined by
$$
\mathcal{U} 
= 
\{(U_k)_{k \in \Z_L} ~|~ U_k \in \Q,
~\sum_{k=1}^L U_k < \frac{L}{2} \},
$$
which includes $\mathcal{U}_{\mathrm{BBS}}$.

\begin{Lemma}\label{lem:P}
Let 
\[
P_{k,m}:=\prod_{l=0}^m\left(\frac{\delta}{u_{k-l}^tu_{k-l-1}^t}\right),\qquad
P:=P_{1,L}=\prod_{l=1}^L\left(\frac{\delta}{u_{l}^tu_{l-1}^t}\right).
\]
When $(U_k^t)_{k\in\Z_L}$ and $(U_k^{t+1})_{k\in\Z_L}$ 
are contained in $\mathcal{U}$,
the periodic discrete KdV $(\ref{KdV})$ is equivalent to:
\begin{gather}\label{KdV2}
u_{k}^{t+1}=\frac{\delta}{u_k^t}\!\left(\!
1\!+\!
\frac{1-P}
{P_{k,0}+P_{k,1}+\cdots+P_{k,L-1}}
\!\right).
\end{gather}
\end{Lemma}

\begin{proof}

Using (\ref{KdV}) recursively, we have
\begin{align*}
u_{k}^{t+1}&=u_{k-1}^t+\frac{\delta}{u_k^t}-\frac{\delta}{u_{k-1}^{t+1}}
=u_{k-1}^t+\frac{\delta}{u_k^t}-\frac{\delta}{u_{k-2}^t+\frac{\delta}{u_{k-1}^t}-\frac{\delta}{u_{k-2}^{t+1}}}\\
&=\cdots =u_{k-1}^t+\frac{\delta}{u_k^t}-\frac{\delta}
{u_{k-2}^t+\frac{\delta}{u_{k-1}^t}-\frac{\delta}{\ddots -\frac{\delta}{u_{k-L}^{t+1}}}}.
\end{align*}
We can regard this continued fraction as a quadratic equation in $u_k^{t+1}$
because of the periodic boundary condition $u_{k-L}^{t+1}=u_k^{t+1}$.
The two solutions of this quadratic equation are expressed as:
\[
u_{k}^{t+1}=\frac{\delta}{u_k^t},\quad\mbox{or}\quad u_k^{t+1}=\frac{\delta}{u_k^t}\!\left(\!
1\!+\!
\frac{1-P}
{P_{k,0}+P_{k,1}+\cdots+P_{k,L-1}}
\!\right).
\]
However, the first solution contradicts the condition 
$(U_{k}^{t})_{k\in \Z_L},(U_{k}^{t+1})_{k\in \Z_L}\in \mathcal{U}$.
In fact, if $u_{k}^{t+1}=\delta/u_k^t$, the following equation should be true:
\[
\sum_{k}{U_k^{t+1}}=\val\,(\prod_k{u_k^{t+1}})=\val\,(\prod_k{\frac{\delta}{u_k^{t}}})
=L-\sum_{k}{U_k^{t}},
\]
which implies $\sum_{k}{U_k^{t}}\geq L/2$ or $\sum_{k}{U_k^{t+1}}\geq L/2$.
\end{proof}

Let
$K_{>0}\subset K$ be the semifield defined by
\[
K_{>0}:=\{c_{-n/d}~t^{-n/d}+c_{(-n+1)/d}~t^{-(n+1)/d}+\cdots\,\vert\,c_{-n/d}>0,~n/d\in\Q_{>0}\}.
\]
For $a,b\in K_{>0}$, it follows that
$\val(a+b)=\min[\val(a),\val(b)]$,
which is not always true on $K$.

Assume $u_k^t\in K_{>0}$ for all $n,t$. By taking the valuation of (\ref{KdV2}), we have
the following proposition:
\begin{Prop}\label{prop:U-evol}
The tropicalization of the periodic discrete KdV equation 
is given by the piecewise-linear map 
\begin{align}\label{pBBS-eq}
U_k^{t+1} 
= 
\min\Bigl[1-U_k^t, 
\max_{m=0,1,\ldots,L-1}\bigl[ \sum_{j=1}^{m+1}{U_{k-j}^t}-\sum_{j=1}^m{(1-U_{k-j}^t})
\bigr]
\Bigr]
\end{align}
on $\mathcal{U}$. We refer to this system as the tropical KdV equation.
This evolution equation is closed 
on the phase space $\mathcal{U}_{\rm BBS}$.
\end{Prop}

\begin{proof}
If $u_k^t\in K_{>0}$ and $(U_k^t)_k\in \mathcal{U}$, we have
\[
\val(P_{k,0}+P_{k,1}+\cdots+P_{k,L-1})
=\min[\val(P_{k,0}),\val(P_{k,1}),\cdots,\val(P_{k,L-1})]
\]
and $\val(1-P)=0$.
Then, direct calculations conclude (\ref{pBBS-eq}) immediately.
Next, we prove that $U_k^t\in\{0,1\}$ implies $U_k^{t+1}\in\{0,1\}$.
For this, it is enough to prove 
\[
N_k^t := 
\max_{m=0,1,\ldots,L-1}\bigl[ \sum_{j=1}^{m+1}{U_{k-j}^t}-\sum_{j=1}^m{(1-U_{k-j}^t})
\bigr]\geq 0.
\]
It is easily checked that 
$N_k^t$
equals to the number of arcs (\S 4.1) which straddle the boundary between 
the $(k-1)$-th box and the $k$-th box at time $t$.
Especially, this number is non-negative.
\end{proof}

\begin{Cor}
Let
\[
U_k^t:=\{\mbox{the number of ball in the $n$-th box at time $t$}\}\in\{0,1\}.
\]
Then, the evolution equation $(\ref{pBBS-eq})$ is equivalent to the pBBS.
\end{Cor}

\begin{proof}
By the definition of pBBS introduced in \S 4.1, we find 
\[
\begin{array}{l}
U_k^{t+1}=1 \iff U_k^t=0 \mbox{ and } N_k^t>0,\\[1mm]
U_k^{t+1}=0 \iff U_k^t=1 \mbox{ or } N_k^t=0,
\end{array}
\]
which is equivalent to $U_{k+1}^{t+1}=\min[1-U_k^t,N_k^t]$.
\end{proof}

We obtain the following as a corollary of Lemma \ref{lem:KdV-curve}.

\begin{Cor}
Let $M$ be 
the maximum integer satisfying $M < \frac{L}{2}$.
The tropical spectral curve $\Gamma$ of 
the tropical KdV equation is given by 
\begin{align}\label{U-spectral}
F(X,Y) =
\min\bigl[2Y, Y + \min_{j=0,\ldots,M}[jX+C_j], XL, L\bigr].
\end{align}
Here $C_j ~(j=1,\cdots,M)$ are tropical functions on 
$\mathcal{U}$.
\end{Cor}

\subsection{Tropical spectral curve of pBBS}
\label{subsec:4-3}

The main result of this subsection is the following.

\begin{Prop}\label{prop:pBBS-curve}
Fix the conserved quantity of the pBBS as
$\lambda=(\lambda_1 < \cdots < \lambda_g)$.
Then the corresponding spectral curve $\Gamma$ of 
the piecewise-linear map \eqref{pBBS-eq} on $\mathcal{U}_{{\rm BBS},\lambda}$ 
is given by 
\begin{align}\label{pBBS-trop}
F(X,Y) = \min[2Y, Y + \min[gX, (g-1)X+C_{g-1},\ldots, X+C_1,C_0], XL, L].
\end{align}
where $C_j = \sum_{i=1}^{g-j} \lambda_i$ for $j=0,\ldots,g-1$.
\end{Prop}

\begin{figure}
\unitlength=1.0mm
\begin{picture}(140,88)(-25,0)

\put(20,10){\vector(1,0){70}}
\put(0,10){\line(1,0){20}}
\put(91,9){\scriptsize $X$}
\put(20,0){\vector(0,1){85}}
\put(16,82){\scriptsize $Y$}

\thicklines

\put(20,6){\scriptsize $0$}
\put(30,10){\line(0,-1){1}}
\put(29,6){\scriptsize $1$}
\put(40,10){\line(0,-1){1}}
\put(39,6){\scriptsize $\lambda_1$}
\put(50,10){\line(0,-1){1}}
\put(49,6){\scriptsize $\lambda_{2}$}
\put(63,10){\line(0,-1){1}}
\put(62,6){\scriptsize $\lambda_{g-1}$}
\put(73,10){\line(0,-1){1}}
\put(72,6){\scriptsize $\lambda_{g}$}

\put(20,10){\line(-1,-3){3}}
\put(20,10){\line(1,1){20}}
\put(40,30){\line(2,1){10}}
\put(50,35){\line(3,1){4}}
\put(73,39){\line(-5,-1){12}}
\put(73,39){\line(1,0){15}}

\put(29,64){\line(0,1){20}}
\put(29,64){\line(1,-1){11}}
\put(40,53){\line(2,-1){10}}
\put(50,48){\line(3,-1){4}}
\put(73,44){\line(-5,1){12}}
\put(73,44){\line(1,0){15}}

\put(20,10){\line(1,6){9}}
\put(40,30){\line(0,1){23}}
\put(50,35){\line(0,1){13}}
\put(63,37){\line(0,1){9}}
\put(73,39){\line(0,1){5}}
\put(54,41){$\cdots$}

\thinlines 

\put(33,42){\oval(6,27)}
\put(35,41){\scriptsize$\wedge$}
\put(31,36){\scriptsize $B_1$}
\put(45,42){\oval(6,10)}
\put(47,41){\scriptsize$\wedge$}
\put(43,40){\scriptsize $B_2$}
\put(68,41.5){\oval(5,5)}
\put(69.5,40.5){\scriptsize$\wedge$}
\put(67,47){\scriptsize $B_g$}

\end{picture}
\caption{Spectral curve for pBBS}
\end{figure}
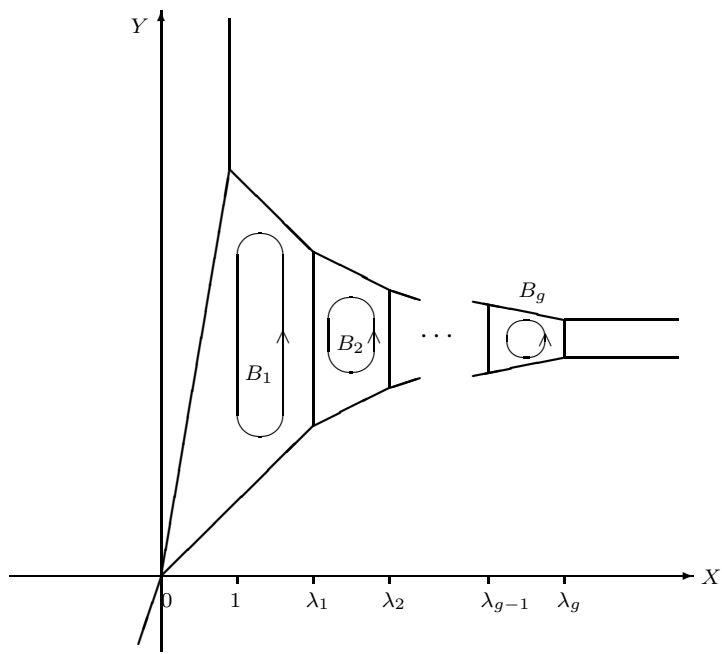

See Figure 4 for the tropical curve $\Gamma$.

To prove Proposition \ref{prop:pBBS-curve},
we have to describe the detail of $C_j$
by applying \cite{MIT05}.
Define $\mathcal{L} = \{1,\ldots,L\}$ and 
$\mathcal{I}_j = \{ I \subset \mathcal{L} ~|~ |I|=j, 
~\text{no consecutive numbers in } I\}$.
Consider a $2 \times L$ lattice and fix $I \in \mathcal{I}_j$.
On the lattice we set $\circ$ (resp. $\bullet$) at  
the top and bottom of the $k$-th column for $k \in I$ 
(resp. $k \in \mathcal{L} \setminus I$).
Here is the example of $L=9$ and $I=\{2,4,7\} \in \mathcal{I}_3$:

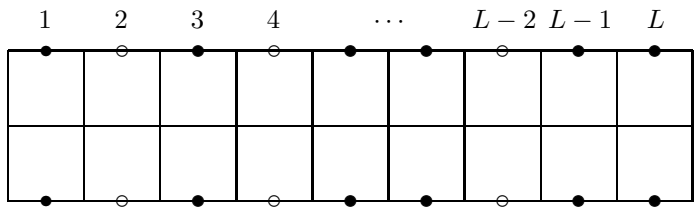
\begin{center}
\unitlength=1.0mm
\begin{picture}(100,35)(0,25)

\put(4,53){\small{$1$}}
\put(14,53){\small{$2$}}
\put(24,53){\small{$3$}}
\put(34,53){\small{$4$}}
\put(48,53){\small{$\cdots$}}
\put(61,53){\small{$L-2$}}
\put(71,53){\small{$L-1$}}
\put(84,53){\small{$L$}}

\put(0,50){\line(1,0){90}}  
\put(0,40){\line(1,0){90}}  
\put(0,30){\line(1,0){90}}  
\multiput(0,50)(10,0){10}{\line(0,-1){20}}

\multiput(5,50)(20,0){3}{\circle*{1.5}}
\multiput(5,30)(20,0){3}{\circle*{1.5}}
\multiput(75,30)(0,20){2}{\circle*{1.5}}
\multiput(55,30)(0,20){2}{\circle*{1.5}}
\multiput(85,30)(0,20){2}{\circle*{1.5}}
\multiput(15,50)(10,0){8}{\circle{1.5}}
\multiput(15,30)(10,0){8}{\circle{1.5}}

\end{picture}
\end{center}
We identify the two vertical boundaries, and 
tile this lattice with the following $4$ patterns:

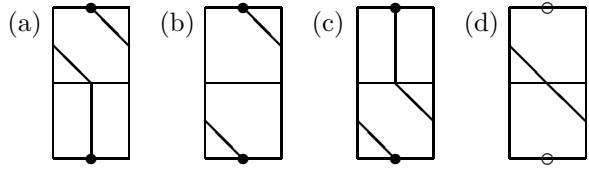
\begin{center}
\unitlength=1.0mm
\begin{picture}(100,30)(0,10)
\multiput(0,35)(20,0){4}{\line(0,-1){20}}
\multiput(10,35)(20,0){4}{\line(0,-1){20}}
\multiput(0,35)(20,0){4}{\line(1,0){10}}
\multiput(0,25)(20,0){4}{\line(1,0){10}}
\multiput(0,15)(20,0){4}{\line(1,0){10}}

\multiput(5,35)(20,0){3}{\circle*{1.5}}
\multiput(5,15)(20,0){3}{\circle*{1.5}}
\multiput(65,35)(0,-20){2}{\circle{1.5}}

\thicklines

\multiput(5,35)(20,0){2}{\line(1,-1){5}}
\multiput(20,20)(20,0){2}{\line(1,-1){5}}
\put(0,30){\line(1,-1){5}}\put(5,25){\line(0,-1){10}}
\put(45,35){\line(0,-1){10}}\put(45,25){\line(1,-1){5}}
\put(60,30){\line(1,-1){10}}

\put(-6,32){\small{(a)}}
\put(14,32){\small{(b)}}
\put(34,32){\small{(c)}}
\put(54,32){\small{(d)}}
\end{picture}
\end{center}
in such a way that (a)-(c) are located at $k \in \mathcal{L} \setminus I$
and (d) is at $k \in I$ to have just $L-j$ non-intersecting paths
each of which starts from upper $\bullet$ and ends at lower $\bullet$.
Note that $\mathcal{I}_j = \emptyset$ for $j > M$, 
and that no tilling exists when $I$ includes consecutive numbers.
The following is one of the possible tillings of the above example:

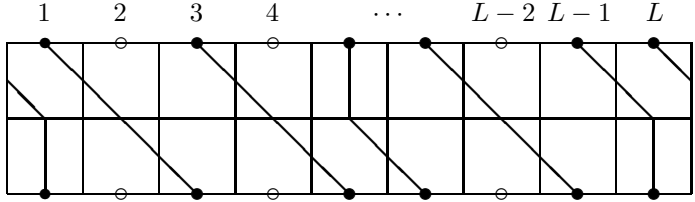
\begin{center}
\unitlength=1.0mm
\begin{picture}(100,35)(0,25)

\put(4,53){\small{$1$}}
\put(14,53){\small{$2$}}
\put(24,53){\small{$3$}}
\put(34,53){\small{$4$}}
\put(48,53){\small{$\cdots$}}
\put(61,53){\small{$L-2$}}
\put(71,53){\small{$L-1$}}
\put(84,53){\small{$L$}}

\put(0,50){\line(1,0){90}}  
\put(0,40){\line(1,0){90}}  
\put(0,30){\line(1,0){90}}  
\multiput(0,50)(10,0){10}{\line(0,-1){20}}

\multiput(5,50)(20,0){3}{\circle*{1.5}}
\multiput(5,30)(20,0){3}{\circle*{1.5}}
\multiput(75,30)(0,20){2}{\circle*{1.5}}
\multiput(55,30)(0,20){2}{\circle*{1.5}}
\multiput(85,30)(0,20){2}{\circle*{1.5}}
\multiput(15,50)(10,0){8}{\circle{1.5}}
\multiput(15,30)(10,0){8}{\circle{1.5}}


\thicklines

\multiput(5,50)(20,0){2}{\line(1,-1){20}}
\put(55,50){\line(1,-1){20}}
\put(45,50){\line(0,-1){10}}\put(45,40){\line(1,-1){10}}
\put(75,50){\line(1,-1){10}}\put(85,40){\line(0,-1){10}}
\put(85,50){\line(1,-1){5}}
\put(0,45){\line(1,-1){5}}\put(5,40){\line(0,-1){10}}
\end{picture}
\end{center}
We write $F_I$ for such a tilling, and call $F_I$ a {\it possible tilling}
for $I$. 
Define functions on $\mathcal{U}$ by 
\begin{align}\label{eq:xi}
\xi(U_k;F_I) = 
\begin{cases} 
1- U_k & \text{ if $F_I$ has (c) at $k$-th column}
\\ 
U_k & \text{ if $F_I$ has (a) at $k$-th column}
\\
0 &  \text{ otherwise,}
\end{cases}
\qquad k \in \mathcal{L}.
\end{align}
For instance, the above tilling $F_I$ gives
$\xi(U_1;F_I) = U_1$, $\xi(U_5;F_I) = 1-U_5$, $\xi(U_9;F_I) = U_9$
and $\xi(U_k;F_I)=0$ for other $k$.      
We remark $\xi(U_k;F_I) \in \{0,1\}$ 
on $\mathcal{U}_{\text{BBS}} \subset \mathcal{U}$.

As a tropicalization of \cite[Th. II.1]{MIT05}, we obtain the following:
\begin{Lemma}
The conserved quantities $C_j$ are written as 
\begin{align}\label{eq:Cj}
C_j 
= 
\min_{I \in \mathcal{I}_j} \min_{F_I} 
\left[\sum_{k \in \mathcal{L}} \xi(U_k;F_I)\right]
\qquad  j=0,\ldots,M.
\end{align}
\end{Lemma}
Here is a key lemma for Proposition \ref{prop:pBBS-curve}:
\begin{Lemma}\label{lemma:C}
Fix $\lambda=(\lambda_1<\lambda_2<\cdots<\lambda_g)$.
On $\mathcal{U}_{{\rm BBS},\lambda} \subset \mathcal{U}_{\text{BBS}}$, 
$C_j$ \eqref{eq:Cj} are written as follows:
\\
(i) $C_j = \sum_{i=1}^{g-j} \lambda_j$ for $i=0,\ldots,g-1$,
and $C_g = 0$.
\qquad 
(ii) $C_j = 0$ for $g < j \leq M$.
\end{Lemma}

\begin{proof}
We are to show the tilling $F_I$ such that 
$C_j = \sum_{k \in \mathcal{L}} \xi(U_k;F_I)$
explicitly.
The point is to find a tilling $F_I$ which realizes
$\xi(U_k;F_I)=0$ for as many $k \in \mathcal{L}$ as possible.

(i) When $j=0$, $\mathcal{I}_0$ is empty and the $2 \times N$ lattice 
does not have $\circ$. Thus a possible tilling is 
given by filling all lattices with (a) or with (c).
Since $|\lambda| < \frac{L}{2}$, the tilling with (a) as 
\begin{center}
\unitlength=0.9mm
\begin{picture}(140,27)(0,28)

\put(-8,53){\small{$U_k$:}}
\multiput(4,53)(10,0){3}{\small{$0$}}
\multiput(34,53)(10,0){2}{\small{$1$}}
\put(53,53){\small{$\cdots$}}
\put(64,53){\small{$1$}}
\multiput(74,53)(10,0){3}{\small{$0$}}
\multiput(104,53)(20,0){2}{\small{$1$}}
\put(113,53){\small{$\cdots$}}
\put(134,53){\small{$0$}}

\put(0,50){\line(1,0){140}}
\put(0,40){\line(1,0){140}}
\put(0,30){\line(1,0){140}}
\multiput(0,50)(10,0){15}{\line(0,-1){20}}

\multiput(5,50)(10,0){14}{\circle*{1.5}}
\multiput(5,30)(10,0){14}{\circle*{1.5}}

\thicklines

\put(0,45){\line(1,-1){5}}\put(5,40){\line(0,-1){10}}
\multiput(5,50)(10,0){13}{\line(1,-1){10}}
\multiput(15,40)(10,0){13}{\line(0,-1){10}}
\put(135,50){\line(1,-1){5}}

\end{picture}
\end{center}
gives the minimum and $C_0 = |\lambda|$ is realized.
We write $F_0$ for the above tilling.

Let us show the cases of $j=1,\ldots,g$.
Due to Remark \ref{rem:Tvm} and Lemma \ref{lem:U},
it is enough to consider the state without soliton scattering,
since  $C_j$'s are conserved by the evolutions $T_m$.
Set  
\begin{align}\label{eq:n}
n_i = \text{``the coordinate of the soliton of length $\lambda_i$"} 
\in \Z / L\Z,
\end{align}
and define $I_j = \{n_i-1 ~(i=g-j+1, \ldots,g) \} \in \mathcal{I}_j$.
Define a tilling $F_j := F_{I_j}$ by replacing the tiles of $F_0$ 
at $n_i -1 \leq k \leq n_i+\lambda_i~(n_i-1 \in I_j)$
with the tiles (b)-(d) as follows:
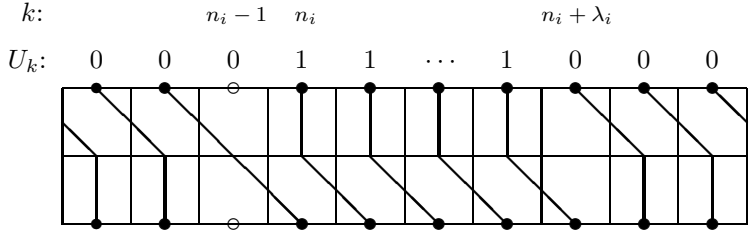
\begin{center}
\unitlength=0.9mm
\begin{picture}(140,35)(0,28)

\put(-6,60){\small{$k$:}}
\put(21,60){\scriptsize{$n_i-1$}}
\put(34,60){\scriptsize{$n_i$}}
\put(70,60){\scriptsize{$n_i+\lambda_i$}}

\put(-8,53){\small{$U_k$:}}
\multiput(4,53)(10,0){3}{\small{$0$}}
\multiput(34,53)(10,0){2}{\small{$1$}}
\put(53,53){\small{$\cdots$}}
\put(64,53){\small{$1$}}
\multiput(74,53)(10,0){3}{\small{$0$}}

\put(0,50){\line(1,0){100}}
\put(0,40){\line(1,0){100}}
\put(0,30){\line(1,0){100}}
\multiput(0,50)(10,0){11}{\line(0,-1){20}}

\multiput(5,50)(0,-20){2}{\circle*{1.5}}
\multiput(15,50)(0,-20){2}{\circle*{1.5}}
\multiput(35,50)(10,0){7}{\circle*{1.5}}
\multiput(35,30)(10,0){7}{\circle*{1.5}}
\multiput(15,50)(10,0){9}{\circle{1.5}}
\multiput(15,30)(10,0){9}{\circle{1.5}}

\thicklines

\put(0,45){\line(1,-1){5}}\put(5,40){\line(0,-1){10}}
\put(5,50){\line(1,-1){10}}\put(15,40){\line(0,-1){10}}
\put(15,50){\line(1,-1){20}}
\put(35,50){\line(0,-1){10}}\put(35,40){\line(1,-1){10}}
\put(45,50){\line(0,-1){10}}\put(45,40){\line(1,-1){10}}
\put(55,50){\line(0,-1){10}}\put(55,40){\line(1,-1){10}}
\put(65,50){\line(0,-1){10}}\put(65,40){\line(1,-1){10}}
\put(75,50){\line(1,-1){10}}\put(85,40){\line(0,-1){10}}
\put(85,50){\line(1,-1){10}}\put(95,40){\line(0,-1){10}}
\put(95,50){\line(1,-1){5}}

\end{picture}
\end{center}
Then we have 
$$
\xi(U_k;F_j) = 
\begin{cases}
1 & n_i \leq k \leq n_i+\lambda_i-1; n_i-1 \in I_j
\\
0 & \text{otherwise},
\end{cases}
$$
and obtain 
$\sum_{k \in \mathcal{L}} \xi(U_k;F_j) = \sum_{j=1}^{g-k} \lambda_j$.
By the induction on $j$, it is easy to show that this is the minimum, 
and the claim follows.

(ii) It is enough to show that there is $I \in \mathcal{I}_j$ and $F_I$
such that $C_j= \sum_{k \in \mathcal{L}} \xi(U_k;F_I) = 0$
for the state $(U_k) \in \mathcal{U}_{{\rm BBS},\lambda}$ for $j=g+1,\ldots,M$. 
From $j=g+1$ to $M$, such $I$ and $F_I$ are recursively constructed
as follows. 

If there is $m_a \in I_{j-1}$ as 
$I = I_{j-1} \cup \{m_a-2\} \in \mathcal{I}_j$,
define $I_{j} = I$.
Otherwise, there certainly exists 
a subset $J = \{m_a-3,m_a+2n+3,m_a + 2l ~(l=0,\ldots,n)\}$ of 
$I_{j-1}$ for some $m_a \in \mathcal{L}$ and $n$ as $0 \leq n < M$.
Then define $I_j = I_{j-1} \cup \{m_a-3,m_a+2n+3,m_a+2l+1 ~(l=-1,0,\ldots,n) \} \setminus J \in \mathcal{I}_j$.

In the first case of $I_j$,
define a possible tilling $F_j$
by replacing the tiles of $F_{j-1}$ at $k=m_a-1$ and $k=m_a-2$ 
with (b) and (d) respectively
(if needed, replace the tile at $k=m_a-3$ with (b)).
Since these replacements do not change the quantity of the function $\xi$ 
\eqref{eq:xi}, 
$\xi(U_k;F_g)=\xi(U_k;F_{g+1})=0$ for $k=m_a-1,m_a-2$ $(m_a-3)$, 
$C_j = 0$ follows from $C_{j-1}=0$.

In the second case of $I_j$, define a possible tilling $F_j$
by replacing the tiles of $F_{j-1}$ at 
$k \in \{m_a+2l+1 ~(l=-1,0,\ldots,n) \}$ 
with (d) and 
the tiles at $k \in \{m_a+2l ~(l=-1,0,\ldots,n+1) \}$ with (b).
These replacements do not change the quantity of the function $\xi$,
and $C_j=0$ follows. 
(See the following example.) 
\end{proof}

\begin{Example}
The case of $L=8$, $\lambda=(1,2)$ and $(n_1,n_2)=(3,6)$.
Here is the tilling $F_2$ with $I_2 = \{2,5\} \in \mathcal{I}_2$, 
which gives $C_2 = 0$:
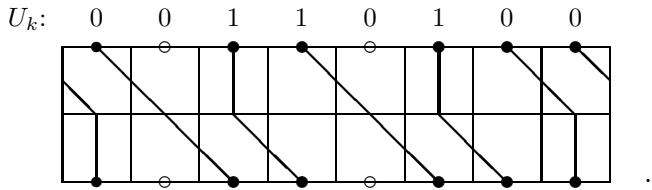
\begin{center}
\unitlength=0.9mm
\begin{picture}(140,30)(0,28)

\put(-8,53){\small{$U_k$:}}
\multiput(4,53)(10,0){2}{\small{$0$}}
\multiput(24,53)(10,0){2}{\small{$1$}}
\put(44,53){\small{$0$}}
\put(54,53){\small{$1$}}
\multiput(64,53)(10,0){2}{\small{$0$}}

\put(0,50){\line(1,0){80}}
\put(0,40){\line(1,0){80}}
\put(0,30){\line(1,0){80}}
\multiput(0,50)(10,0){9}{\line(0,-1){20}}

\multiput(5,50)(0,-20){2}{\circle*{1.5}}
\multiput(25,50)(0,-20){2}{\circle*{1.5}}
\multiput(35,50)(0,-20){2}{\circle*{1.5}}
\multiput(55,50)(0,-20){2}{\circle*{1.5}}
\multiput(65,50)(0,-20){2}{\circle*{1.5}}
\multiput(75,50)(0,-20){2}{\circle*{1.5}}
\multiput(15,50)(10,0){7}{\circle{1.5}}
\multiput(15,30)(10,0){7}{\circle{1.5}}

\thicklines

\put(0,45){\line(1,-1){5}}\put(5,40){\line(0,-1){10}}
\put(5,50){\line(1,-1){20}}
\put(25,50){\line(0,-1){10}}\put(25,40){\line(1,-1){10}}
\put(35,50){\line(1,-1){20}}
\put(55,50){\line(0,-1){10}}\put(55,40){\line(1,-1){10}}
\put(65,50){\line(1,-1){10}}\put(75,40){\line(0,-1){10}}
\put(75,50){\line(1,-1){5}}

\put(85,30){.}
\end{picture}
\end{center}
By setting $I_3 = \{2,5,8\} \in \mathcal{I}_3$ ($m_a=2$), 
we can define $F_3$ which gives $C_3=0$ as
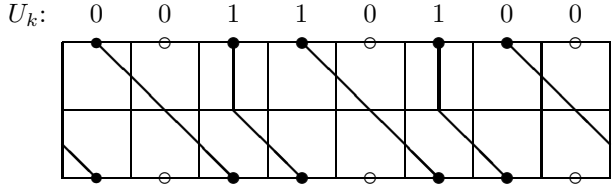
\begin{center}
\unitlength=0.9mm
\begin{picture}(140,30)(0,28)

\put(-8,53){\small{$U_k$:}}
\multiput(4,53)(10,0){2}{\small{$0$}}
\multiput(24,53)(10,0){2}{\small{$1$}}
\put(44,53){\small{$0$}}
\put(54,53){\small{$1$}}
\multiput(64,53)(10,0){2}{\small{$0$}}

\put(0,50){\line(1,0){80}}
\put(0,40){\line(1,0){80}}
\put(0,30){\line(1,0){80}}
\multiput(0,50)(10,0){9}{\line(0,-1){20}}

\multiput(5,50)(0,-20){2}{\circle*{1.5}}
\multiput(25,50)(0,-20){2}{\circle*{1.5}}
\multiput(35,50)(0,-20){2}{\circle*{1.5}}
\multiput(55,50)(0,-20){2}{\circle*{1.5}}
\multiput(65,50)(0,-20){2}{\circle*{1.5}}
\multiput(15,50)(10,0){7}{\circle{1.5}}
\multiput(15,30)(10,0){7}{\circle{1.5}}

\thicklines

\put(0,35){\line(1,-1){5}}
\put(5,50){\line(1,-1){20}}
\put(25,50){\line(0,-1){10}}\put(25,40){\line(1,-1){10}}
\put(35,50){\line(1,-1){20}}
\put(55,50){\line(0,-1){10}}\put(55,40){\line(1,-1){10}}
\put(65,50){\line(1,-1){15}}

\put(85,30){.}
\end{picture}
\end{center}
There is no $m_a \in I_3$ such that $I_3 \cup \{m_a -2\} \in \mathcal{I}_4$.
Thus we set $I_4 = \{2,4,6,8\}$ ($m_a=5,~n=0$) and define $F_4$ as
\begin{center}
\unitlength=0.9mm
\begin{picture}(140,30)(0,28)

\put(-8,53){\small{$U_k$:}}
\multiput(4,53)(10,0){2}{\small{$0$}}
\multiput(24,53)(10,0){2}{\small{$1$}}
\put(44,53){\small{$0$}}
\put(54,53){\small{$1$}}
\multiput(64,53)(10,0){2}{\small{$0$}}

\put(0,50){\line(1,0){80}}
\put(0,40){\line(1,0){80}}
\put(0,30){\line(1,0){80}}
\multiput(0,50)(10,0){9}{\line(0,-1){20}}

\multiput(5,50)(0,-20){2}{\circle*{1.5}}
\multiput(25,50)(0,-20){2}{\circle*{1.5}}
\multiput(45,50)(0,-20){2}{\circle*{1.5}}
\multiput(65,50)(0,-20){2}{\circle*{1.5}}
\multiput(15,50)(10,0){7}{\circle{1.5}}
\multiput(15,30)(10,0){7}{\circle{1.5}}

\thicklines

\put(0,35){\line(1,-1){5}}
\multiput(5,50)(20,0){3}{\line(1,-1){20}}
\put(65,50){\line(1,-1){15}}

\put(85,30){.}
\end{picture}
\end{center}
 
\end{Example}

\begin{proof}(Proposition \ref{prop:pBBS-curve})
Fix $\lambda = (\lambda_1<\cdots<\lambda_g)$.
From Lemma \ref{lemma:C}, the tropical polynomial \eqref{U-spectral}
is written as
$$
F(X,Y) =
\min\bigl[2Y, Y + \min[\min_{j=0,\ldots,g-1}[jX+C_j],
\min_{j=g,\ldots,M}[jX]],XL, L\bigr],
$$
and the corresponding tropical curve $\Gamma$ 
has three infinite domains determined by $F(X,Y)=XL$, $L$ and $2Y$,
which fill the domain $D=\{(X,Y) \in \R^2 ~|~ X<0 \text{ or } Y<0 \}$.
In the rest domain $\R^2 \setminus D$,
we have $\min_{j=g,\ldots,M}[jX] = gX$.
Thus the defining equation of $\Gamma$ can be reduced to \eqref{pBBS-trop}.
\end{proof}

Though $\Gamma$ is not smooth, 
we can calculate its period matrix $\Omega$ in the same way
as Definition \ref{i:def:omega}:
\begin{align}\label{eq:omega}
\Omega = (\langle \tilde{B}_k, \tilde{B}_j \rangle)_{k,j = 1,\ldots,g}.
\end{align}
Here we set $\tilde{B}_k = \sum_{j=1}^{g+1-k} B_j$
by using the basis $B_j$'s of $\pi_1(\Gamma')$ as Figure 4.
Then we obtain our final result:

\begin{Prop}
The period matrix $\Omega$ \eqref{eq:omega}  coincides with 
the period matrix $A$ \eqref{A}.
In particular, $J(L,\lambda)$ \eqref{pBBS-J}
is nothing but the tropical Jacobian $J(\Gamma)$ of $\Gamma$.
\end{Prop}

\begin{Remark}
The trop-pToda and the pBBS are closely related dynamical systems \cite{IT08}.
Here we just note the relation between 
the two tropical Jacobians studied in \S 3 and \S 4.
By setting $N=g+1$, $C_1=0$ and $C_{N+1} = L$ in \S 3,
and identifying $\lambda_i$'s in the both sections, 
$J(L,\lambda)$ \eqref{pBBS-J} turns out to be isomorphic to the quotient space 
of $J(\Gamma_C)$ \eqref{pToda-J} by the action 
$\nu: J(\Gamma_C) \to J(\Gamma_C);~[z] \mapsto [z + L{\bf e}_1]$.
\end{Remark}


\begin{thebibliography}{99}

\bibitem{DT76}
E.~Date and S.~Tanaka,
Analogue of inverse scattering theory for the discrete Hill's equation 
and exact solutions for the periodic Toda lattice, 
Prog. Theor. Phys., \textbf{55}, 457--465 (1976).

\bibitem{FOY00} 
K.~Fukuda, M.~Okado and Y.~Yamada,
Energy functions in box ball systems,
Internat. J. Modern Phys. A 15, no. 9, 1379--1392 (2000).

\bibitem{H77}
R.~Hirota, 
Nonlinear partial difference equations. I. A difference
analogue of the Korteweg-de Vries equation,
J. Phys. Soc. Japan, \textbf{43}, 1424--1433 (1977).

\bibitem{HT95}
R.~Hirota and S.~Tsujimoto,
Conserved quantities of a class of nonlinear difference-difference 
equations,
J. Phys. Soc. Japan, \textbf{64}, No. 9, 3125--3127 (1995).

\bibitem{IKT11}
R.~Inoue, A.~Kuniba and T.~Takagi,
Integrable structure of Box-ball system:
crystal, Bethe ansatz, ultradiscretization and tropical geometry 
(review article),
J. Phys. A: Math. Theor., \textbf{45}, 073001 (2012).

\bibitem{IT08}
R.~Inoue and T.~Takenawa, 
Tropical spectral curves and integrable cellular automata,
Int. Math. Res. Not. IMRN,  \textbf{27}, Art ID. rnn019, 27 pp (2008).

\bibitem{IT09}
R.~Inoue and T.~Takenawa, 
A tropical analogue of Fay's  trisecant identity and 
an ultra-discrete periodic Toda lattice,
Comm. Math. Phys., \textbf{289}, pp 995--1021 (2009).

\bibitem{IT09b}
R.~Inoue and T.~Takenawa, 
Tropical Jacobian and the generic fiber of the ultra-discrete 
periodic Toda lattice are isomorphic,
RIMS K\^oky\^uroku Bessatsu, B13, pp 175--190 (2009).

\bibitem{IMS-Book}
I.~Itenberg, G.~Mikhalkin and E.~Shustin,
Tropical algebraic geometry,
Oberwolfach Seminars, 35 (Birkh\"user Verlag, Basel, 2007).

\bibitem{KvM75}
M.~Kac and P.~van Moerbeke,
On some periodic Toda lattices,
Proc. Natl. Acad. Sci. USA \textbf{72}, 1627--1629 (1975);
A complete solution of the periodic Toda problem,
id., 2879--2880 (1975).

\bibitem{KimijimaTokihiro02}
T.~Kimijima and T.~Tokihiro,
Initial-value problem of the discrete periodic 
Toda equations and its ultradiscretization,
Inverse Problems, \textbf{18}, 1705--1732 (2002).
 
\bibitem{MikhaZhar06}
G.~Mikhalkin and I.~Zharkov, 
Tropical curves, their Jacobians and theta functions,
Contemp. Math., \textbf{465}, 203--230 (2008).

\bibitem{KS06}
A. Kuniba and R. Sakamoto,
The Bethe ansatz in a periodic box-ball system and 
the ultradiscrete Riemann theta function,  
J. Stat. Mech., P09005 (2006). 

\bibitem{KTT06}
A.~Kuniba, T.~Takagi and A.~Takenouchi, 
Bethe ansatz and inverse scattering transform in a periodic box-ball system,
Nuclear Phys. B 747, no. 3, 354--397 (2006).

\bibitem{MIT05}
J.~Mada, M.~Idzumi and T.~Tokihiro,
Path description of conserved quantities of generalized periodic 
box-ball systems,
J. Math. Phys., \textbf{46} 022701 (2005).

\bibitem{MIT08}
J.~Mada, M.~Idzumi and T.~Tokihiro,
The box-ball system and the $N$-soliton solution of the 
ultradiscrete KdV equation,
J. Phys. A: Math. Theor., \textbf{41}, 175207 (23pp) (2008).

\bibitem{TS90}
D.~Takahashi and J.~Satsuma,
A soliton cellular automaton,
J. Phys. Soc. Japan, \textbf{59}, no. 10, 3514--3519 (1990).

\bibitem{YT02}
F.~Yura and T.~Tokihiro,
On a periodic soliton cellular automaton,
J. Phys. A: Math. Gen., \textbf{35}, 3787--3801 (2002).

\end{thebibliography}
\end{document}